\documentclass[12pt,a4paper]{article}

\usepackage[margin=1in]{geometry} %page size
\usepackage{amsfonts,amscd,amssymb, mathtools,mathrsfs, dsfont, bbm,bbding} %math and other symbols
\usepackage[amsmath,amsthm,thmmarks]{ntheorem} %theorem style, auto qedhere
\usepackage{graphicx,xypic,color,float} %graph
\usepackage{indentfirst}%noindent
\usepackage{lmodern}%better output
\usepackage[T1]{fontenc} %better output
\usepackage{enumerate,listings,verbatim,paralist}%list%enumitem
\usepackage{setspace,xspace}%set space, shortward command
\usepackage[colorlinks=true,citecolor=blue]{hyperref}
\usepackage{extarrows}%used to define def
\usepackage{pdfpages}%insert pdf files
\usepackage{ulem}

\usepackage{setspace}
\AtBeginDocument{%
\addtolength\abovedisplayskip{-0.15\baselineskip}%
\addtolength\belowdisplayskip{-0.15\baselineskip}%
\addtolength\abovedisplayshortskip{-0.15\baselineskip}%
\addtolength\belowdisplayshortskip{-0.15\baselineskip}%
}

\newtheorem{theorem}{Theorem} [section]
\newtheorem{proposition}[theorem]{Proposition}
\newtheorem{corollary}[theorem]{Corollary}
\newtheorem{lemma}[theorem]{Lemma}
\newtheorem{definition}{Definition}[section]
\newtheorem{assumption}{Assumption}[section]
\newtheorem{remark}{Remark}[section]

\numberwithin{equation}{section} %give equations number

\renewcommand{\geq}{\geqslant}
\renewcommand{\leq}{\leqslant}

\newcommand{\citethm}[1]{Theorem \ref{#1}}
\newcommand{\citeprop}[1]{Proposition \ref{#1}}
\newcommand{\citecoro}[1]{Corollary \ref{#1}}
\newcommand{\citelem}[1]{Lemma \ref{#1}}

\newcommand{\citeassmp}[1]{Assumption \ref{#1}}

\newcommand{\opfont}{\mathbb}

\newcommand{\BE}[2][]{\ensuremath{\operatorname{\opfont{E}}^{#1}\!\left[#2\right]}}
\newcommand{\bp}{\ensuremath{\opfont{P}}}
\newcommand{\BP}[2][]{\ensuremath{\operatorname{\opfont{P}}^{#1}\!\left(#2\right)}}
\newcommand{\BF}{\ensuremath{\mathcal{F}}}

\newcommand{\R}{\ensuremath{\operatorname{\mathbb{R}}}}

\newcommand{\dd}{\ensuremath{\operatorname{d}\! }}
\newcommand{\dt}{\ensuremath{\operatorname{d}\! t}}
\newcommand{\ds}{\ensuremath{\operatorname{d}\! s}}

\newcommand{\setq}{\mathscr{Q}}

\newcommand{\idd}[1]{\ensuremath{\operatorname{\mathds{1}}_{#1}}}

\newcommand{\nn}{\nonumber}

\newcommand{\einf}{\ensuremath{\mathrm{ess\:inf\:}}}
\newcommand{\esup}{\ensuremath{\mathrm{ess\:sup\:}}}
\newcommand{\barq}{\overline{Q}}
\newcommand{\ep}{\varepsilon}
\newcommand{\opl}{\Phi} 
\newcommand{\opll}{\Psi} 
\newcommand{\dsp}{\displaystyle}
\newcommand{\setrho}{\mathscr{U}_{\rho}}

\newcommand{\seta}{\mathscr{A}}
\def\game{Q_{\pr}}
\def\pr{\vartheta}

\def\tr{\intercal}
\newcommand{\lam}{\lambda}
\newcommand{\llam}{\underline{\lambda}}
\newcommand{\veum}{V_{\mathrm{\tiny{EU}}}}
\newcommand{\setcb}{C_{\mathrm{\tiny{RCI}}}^{1-}((0,1))}

\begin{document}
%%%%%%%%%%%%%%%%%%%%%%%%%%%%%%%%%%%%%%%%%%%%%%%%
\title{Robust utility maximization with intractable claims}
\author{Yunhong Li\thanks{Department of Applied Mathematics, The Hong Kong Polytechnic University, Kowloon, Hong Kong, China. Email: \url{yunhong.li@connect.polyu.hk}.}
\and Zuo Quan Xu\thanks{Department of Applied Mathematics, The Hong Kong Polytechnic University, Kowloon, Hong Kong, China. Email: \url{maxu@polyu.edu.hk}. }
\and Xun Yu Zhou\thanks{Department of Industrial Engineering and Operations Research \& The Data Science Institute, Columbia University, New York, NY 10027, USA. Email: \url{xz2574@columbia.edu}.}
}
\date{\today}
\maketitle

\begin{abstract}
We study a continuous-time expected utility maximization problem in which the investor at maturity receives the value of a contingent claim in addition to the investment payoff from the financial market. The investor knows nothing about the claim other than its probability distribution, hence an ``intractable claim''.
In view of the lack of necessary information about the claim, we consider a robust formulation to maximize her utility in the worst scenario. We apply the quantile formulation to solve the problem, expressing the quantile function of the optimal terminal investment income as the solution of certain variational inequalities of ordinary differential equations and obtaining the resulting optimal trading strategy. In the case of an exponential utility, the problem reduces to a (non-robust) rank--dependent utility maximization with probability distortion whose solution is available in the literature. The results can also be used to determine the utility indifference price of the intractable claim.


\end{abstract}

\section{Introduction}
The expected utility models have taken a central position in modern financial portfolio selection theory.
In a typical expected utility model, an investor looks for the best portfolio in a given financial market to maximize her utility, based on a dynamically expanding body of information and subject to various constraints. There are also the so-called partial information models in which not all the information are available to the investor and she needs to continuously observe some related information processes to gain better understanding of the underlying market over time, using techniques such as filtering theory and adaptive controls.

In practice, however, there are situations where people, facing future contingent claims (incomes or liabilities) that cannot be understood even if the time is very close to maturity, aim to maximize the utility of a total sum including payoffs from the stock investment as well as those ``external'' claims. For instance, the claim of an insurance contract is not revealed until the related event occurs, the down payment of a mortgage is only known when the purchase has been made, and the value of the employer's stock option may have little to do with the general stock market.
These expected utility models are essentially different from the typical partial information models mentioned above because, while the probability distributions of these claims can be estimated, little or no information on their {\it correlations} with the financial market are available. Such a claim, about which we know nothing but its own distribution, is termed an ``intractable claim'' in Hou and Xu \cite{HX16}.

Because intractable claims are not predictable nor hedgeable until they are realized, classical stochastic control including filtering theory is not applicable.
Hou and Xu \cite{HX16} formulate and investigate a continuous-time Markowitz's mean--variance model with an intractable claim and the objective to minimize the robust variance in the worst scenario.
The special structure of the mean--variance setting allows them to use the completion-of-squares and quantile optimization techniques to derive the optimal solution.

In this paper we consider a continuous-time expected utility model (i.e. Merton's problem) with an intractable claim.
The usual maximization of expected utility is not even well defined because one cannot compute the expected utility of the sum of the investment payoff form the market and the intractable claim without knowing the joint distribution of the two. Instead,
we introduce a robust model whose objective is to maximize in the worst scenario, namely, to maximize the smallest possible expected utility of the total terminal income over all possible realizations of the intractable claim.
While this formulation is inspired by \cite{HX16}, our model loses the mean--variance structure so the approach there largely fails to work here.

At a first glance, the robust model still involves the expected utility of two random payoffs whose correlation is unknown so the aforementioned difficulty seems to be intact. However, we prove that at the {\it minimum} expected utility the two random incomes can be separated through the sum of the (marginal) quantile functions of the two; hence the robust measure is computable based on the {\it respective} information about the market and the claim.

The idea of taking the quantile functions, instead of random variables, as the decision variable to study preference functionals goes back to at least Schied \cite{S04}. Jin and Zhou \cite{JZ08} employ the quantile method to overcome the difficulty arising from probability distortion in a continuous-time behavioral portfolio selection model featuring cumulative prospect theory. He and Zhou \cite{HZ11} put forward a unified theory on this approach that can be used to solve both expected and non-expected utility models and coin the term ``quantile formulation''. The theory has been further developed since then; see e.g. \cite{XZ13,X14,HJZ15,X16,HX16,XZ16,X21,X23}.

In this paper, we apply the quantile formulation to tackle our problem. From the ``first-order condition'' of the corresponding calculus of variations for quantile, we derive an ordinary integro--differential equation (OIDE) as well as a variational inequality type of ordinary differential equations (ODEs) that are satisfied by the quantile function of the optimal terminal payoff from the stock market.
Surprisingly, if the utility is exponential, we show that our robust model can be reduced to a behavioral rank--dependent utility maximization problem whose solution is available (\cite{XZ16}).

The rest of the paper is organized as follows. In Section \ref{secpf}, we introduce the financial market and formulate the robust expected utility maximization (EUM) problem. We turn the problem into its quantile formulation and study its well-posedness in Section \ref{secqf}. Section \ref{secqfsol} presents the solution. In Section \ref{expsp} we discuss the special case of the exponential utility. Finally, Section \ref{conclude} concludes, in which we also discuss how to apply our results to determine the utility indifference price of an intractable claim.

\section{Problem Formulation}\label{secpf}

Throughout this paper, we fix a probability space $(\Omega,\mathbf{F},\bp)$ satisfying the usual assumptions, along with a standard $m$-dimensional Brownian motion $$W=\{(W_1(t), \cdots, W_m(t))^{\tr},\;t\geq0\}$$ representing the uncertainties/risks of the financial market under consideration (to be described below), where (and hereafter) ${A}^\tr$ denotes the transpose of a matrix or vector $A$. Unless otherwise stated,
a r.v. is the shorthand for an $\mathbf{F}$-measurable random variable. We also fix an investment horizon $[0,T]$ where $T>0$ is a constant. Let $\{\BF_t\}_{t\geq0}$ be the filtration generated by $W$ complemented by all the $\bp$-null sets,
and $L^{0}_{\BF_{T}}$ be the set of $\BF_{T}$-measurable random variables.
We stress that $\BF_{T}\subsetneqq\mathbf{F}$; so there is randomness {\it outside} of the financial market or, equivalently, not every random variable belongs to $L^{0}_{\BF_{T}}$.

For any r.v. $Y$, let $F_{Y}$ denote its probability distribution function,
and $Q_{Y}$ its quantile (function) defined by
\[Q_{Y}(t):=\inf\big\{z\in\R\;\big|\; F_{Y}(z)> t\big\}, \quad t\in (0,1), \]
which is the right-continuous inverse function of $F_{Y}$.
It is easy to check that every quantile is a right-continuous and increasing\footnote{In this paper by ``increasing'' we mean ``non-decreasing'', and by ``decreasing'' we mean ``non-increasing''.} (called RCI from now on) function $:(0,1)\to \R$.
On the other hand, every RCI function $Q: (0,1)\to \R$ is the quantile of the r.v. $Q(U)$, where $U$ is any r.v. that is uniformly distributed on $(0,1)$. Hence, the set of quantiles is the same as that of RCI functions $: (0,1)\to \R$.
We may also use the convention that
\[Q_{Y}(0):=\lim_{t\to 0+}Q_{Y}(t) \quad \mbox{and}\quad Q_{Y}(1):=\lim_{t\to 1-}Q_{Y}(t); \]
so
\[Q_{Y}(0)=\einf Y\quad \mbox{and}\quad Q_{Y}(1)=\esup Y.\]
We write $X\sim Y$ if two r.v.'s $X$ and $Y$ follow the same distribution.

\par
Let $\setcb$ be the set of functions $f: (0,1)\to \R$ that are absolutely continuous with derivatives $f'$ being RCI.
Clearly, any function $f\in\setcb$ is convex; hence the limits
$f(0+)=\lim_{t\to 0+}f(t)$ and $f(1-)=\lim_{t\to 1-}f(t)$ exist (but may be infinity).
\par
In what follows, ``almost everywhere'' (a.e.) and ``almost surely'' (a.s.) may be suppressed for expositional simplicity whenever no confusion might occur. Finally,
we use $\vert M\vert$ to denote $\mathrm{trace}(MM^{\tr})$ for any matrix or vector $M$.

Now we are ready to introduce the financial market and our investment problem.

\subsection{Financial Market Model}
\noindent
Consider a continuous-time arbitrage-free financial market where $m+1$ assets are traded continuously on $[0, T]$.
One of the assets is a bond, whose price $S_0(\cdot)$ evolves according to an ordinary differential equation (ODE):
\begin{equation*}
\begin{cases}
\dd S_0(t)=r(t)S_0(t)\dt, \quad t \in [0, T], \\
\ S_0(0)=s_0>0,
\end{cases}
\end{equation*}
where $r(t) $ is the appreciation rate of the bond at time $t$. The remaining $m$ assets are
stocks, and their prices are modeled by a system of stochastic differential equations (SDEs):
\begin{equation*}
\begin{cases}
\dd S_i(t)=S_i(t)[\beta_i(t)\dt+\sum_{j=1}^m\sigma_{ij}(t)\dd W_{j}(t)], \quad t \in [0, T], \\
\ S_i(0)=s_i>0,
\end{cases}
\end{equation*}
where $\beta_i(t) $ is the appreciation rate of the stock $i$ and
$\sigma_{ij}(t)$ is the volatility coefficient at time $t$.
Denote by $\beta:=\{(\beta_1(t), \cdots, \beta_m(t))^{\tr},\; t \in [0, T]\}$ the appreciation rate vector process and by $\sigma:=\{(\sigma_{ij}(t))_{m\times m},\; t \in [0, T]\}$ the volatility matrix process, along with the
{excess return rate vector process} $B:=\{B(t),\; t \in [0, T]\}$ where
\begin{equation*}
B(t) :=\beta(t)-r(t)\mathbf{1}, \quad t \in [0, T],
\end{equation*}
with $\mathbf{1}=(1, 1, \cdots, 1)^{\tr}$ denoting the $m$-dimensional unit vector.

\par
Following the monograph Karatzas and Shreve \cite{KS98}, we impose the following standing assumptions on the market parameters throughout the paper:
\begin{itemize}
\item The processes $r$, $\beta$ and $\sigma$ are all $\{{\cal F}_t\}_{t\geq0}$-progressively measurable;
\item The process $r$ is essentially bounded and the process $\beta$ satisfies
\begin{equation*}
\int_0^T\vert\beta(t)\vert\dt<\infty,\;\;\mbox{a.s.};
\end{equation*}
\item The process $\sigma$ is invertible;
\item The market price of risk process $\theta:=\{\theta(t),\; t \in [0, T]\}$ defined by
\begin{equation*}
\theta(t) :=\sigma(t)^{-1} B(t), \quad t \in [0, T]
\end{equation*}
satisfies
\begin{equation*}
\int_0^T\vert\theta(t)\vert^2\dt<\infty,\;\;\mbox{a.s.},
\end{equation*}
and $\theta$ is not identical to zero.
\item The positive local martingale
\[ t\mapsto\exp\left(-\frac{1}{2}\int_0^{t} \vert \theta(s)\vert^2 \ds-\int_0^{t}\theta(s)^{\tr}\dd W(s)\right), \quad t \in [0, T],\]
is a true martingale.
\end{itemize}
The last requirement is satisfied if the Novikov condition
$\BE{\exp\left(\frac{1}{2}\int_0^{T} \vert \theta(s)\vert^2 \ds\right)}<\infty$ holds.

The above assumptions ensure that our financial market is ``standard and complete''; refer to \cite[Definition 5.1, page 17]{KS98} and \cite[Theorem 6.6, page 24]{KS98} for the definitions.

\subsection{Investment Problem}
\noindent
There is a small investor (``she'') in the market whose transactions have no influence on the asset prices in the market. She has an initial endowment $x>0$ and invests in the financial market over the time period $[0, T]$.
Denote by $\pi_i(t)$ the total market value of her wealth invested in stock $i$ at time $t$, $ i=1, \cdots, m$. We allow short selling so that $\pi_i(t)$ can take negative values. We also assume that the trading of shares takes place continuously in a self-financing fashion (i.e., there is no consumption or income) and the market is frictionless (i.e., transactions do not incur any fees or costs). The investor's
portfolio (process) is $\pi=\{(\pi_1(\cdot), \cdots, \pi_m(\cdot))^{\tr},\;t\geq0\}$, and the corresponding
wealth process $X^{\pi}$ evolves according to the wealth SDE (see Karatazas and Shreve \cite{KS98}):
\begin{align}\label{eq:state-positive}
\begin{cases}
\dd X^{\pi}(t)=\Big[r(t)X^{\pi}(t)+\pi(t)^{\tr}\sigma(t)\theta(t)\Big]\dt+\pi(t)^{\tr}\sigma(t)\dd W(t), \quad t \in [0, T],\\
\ X^{\pi}(0)=x
\end{cases}
\end{align}
In our model, the investor makes investment decisions based on the information $\{{\cal F}_t\}_{t\geq0}$ from the financial market, which are strictly less than those represented by $\mathbf{F}$. So unlike many existing studies in portfolio selection, ours is a partial information model. The following defines precisely an admissible portfolio.
\begin{definition}
We call $\pi$ an admissible portfolio (or strategy) if
it is an $\{{\cal F}_t\}_{t\geq0}$-progressively measurable process on $[0,T]$ such that
\[\int_0^T\vert\sigma(t)^{\tr}\pi(t)\vert^2\dt<\infty,\;\;\mbox{a.s.}\]
and such that the corresponding wealth process $X^{\pi}$ determined by \eqref{eq:state-positive} is always nonnegative.
\end{definition}
Because the market is free of arbitrage, the last requirement of the definition is equivalent to that $X^{\pi}(T)$ is a nonnegative r.v..

For any admissible portfolio $\pi$, the SDE \eqref{eq:state-positive} admits a unique solution $X^{\pi}$ which is called an admissible wealth process, and $(X^{\pi}, \pi)$ is called an admissible pair. Note that \eqref{eq:state-positive} is linear in $X^{\pi}$ and $\pi$; so all the admissible pairs form a convex set.
Moreover, all the admissible pairs are $\{{\cal F}_t\}_{t\geq0}$-progressively measurable.
From now on, we only consider admissible portfolios unless otherwise specified.

The original, classical expected utility maximization (EUM) problem is formulated as
\begin{align}\label{eut}
\sup\limits_{\pi} &~\BE{u(X^{\pi}(T))}, \\
\mathrm{subject\ to} &~ (X^{\pi}, \pi) \text{ being an admissible pair with $X^{\pi}(0)=x$,} \nn
\end{align}
where $u$ is a given utility function.

In our model, however, in addition to the investment payoff from the financial market,
the investor will also receive the value of a claim $\pr$ (for instance, an excise of her company's stock option, a bonus, the down payment of a mortgage or even a lottery prize) at the maturity date $T$. She knows nothing about this claim -- hence an {\it intractable} claim -- except its {\it own} distribution before $T$. In particular, she does not know the {\it joint} distribution between this claim and the market, nor can she estimate this joint distribution purely based on the market data. Moreover, $\pr$ is $\mathbf{F}$-measurable, not $\BF_{T}$-measurable; so it is not a hedgable contingent claim in the financial market.
As a result, it is not meaningful to consider the utility of the total terminal wealth
\[\BE{u(X^{\pi}(T)+\pr)},\]
because to determine this value the investor needs to know the joint distribution of $X^{\pi}(T)$ and $\pr$.
Inspired by Hou and Xu \cite{HX16}, we consider instead the worst scenario over all the possible realizations of $\pr$, leading to the robust preference
\begin{align} \label{J_0def}
J_0(X):=\inf\limits_{Y\sim \pr}\;\BE{u(X+Y)},
\end{align}
where $X\in L^{0}_{\BF_{T}}$ is the payoff from the financial market, and
the infimum is taken over all $\mathbf{F}$-measurable r.v.'s $Y$ that are distributed the same as $\pr$. It turns out, as will be shown later, $J_0$ does not rely on the joint distribution between $X$ and $\pr$, but on their own distributions.

To summarize, we have the following investment model:
\begin{align}\label{eut+}
V_{0}(x):=&~\sup\limits_{\pi} \inf\limits_{Y\sim \pr}\; \BE{u(X^{\pi}(T)+Y)}, \\
\mathrm{subject\ to} &~ (X^{\pi}, \pi) \text{ being an admissible pair with $X^{\pi}(0)=x$.} \nn
\end{align}
The problem can be regarded as a game between the investor (who chooses the best portfolio $\pi$ and its output $X^{\pi}$ in the financial market) and the nature (who chooses $Y$, the worst realization of $\pr$, part of which may come from outside of the financial market).
In contrast to many classical game models, we can not switch the order of $\sup$ and $\inf$ in the problem \eqref{eut+}. It is not a classical control problem either, as we will show in the subsequent analysis that the objective functional is not the standard expectation in stochastic control theory (see, e.g., Yong and Zhou \cite{YZ99}). On the other hand, the classical EUM problem \eqref{eut} can be regarded as the special case of \eqref{eut+} where $\pr\equiv 0$, in which case the optimal value is denoted by $\veum(x)$.
\par
To ensure the problem \eqref{eut+} to be well-defined and to
avoid undue technicalities, we make the following assumption throughout this paper:
\begin{assumption}\label{boundedclaim}
The claim $\pr$ is almost surely lower bounded, that is,
\[Q_{\pr}(0)=\einf\pr>-\infty.\]
Moreover, the utility function $u$ is continuous on $\big[\min\big\{0,\game(0)\big\},\infty\big)$ and continuously differentiable on $\big(\min\big\{0,\game(0)\big\},\infty\big)$ with $u'$ being strictly decreasing to 0 at infinity.
Furthermore,
\begin{align*}
\BE{u(\pr)}<\infty.
\end{align*}

\end{assumption}

\begin{remark}
The lower boundedness assumption on $\pr$ is very mild and reasonable, satisfied in most practical cases. The last assumption $\BE{u(\pr)}<\infty$ is fulfilled if $\pr$ has a finite expectation, thanks to the concavity of $u$.
\end{remark}

Under Assumption \ref{boundedclaim}, the utility function $u$ is strictly increasing and strictly concave on the interval $[\min\{0,\game(0)\},\infty)$, under which we will show that
the problem \eqref{eut+} is well-posed, namely, it has a finite optimal value (see \citethm{wellposedness} below).

\section{Quantile Formulation}\label{secqf}
\noindent

In this section we turn the problem \eqref{eut+} into a static optimization problem and present its quantile formulation.

\subsection{A static optimization problem}

Employing the well-established martingale method, we can divide solving the problem \eqref{eut+} into two steps. In the first step, we find the optimal solution $X^*$ to the following {\it static} optimization problem
\begin{eqnarray}\label{static1}
\begin{array}{rl}
\sup\limits_{X\in L^{0}_{\BF_{T}}} &\; J_0(X),\\ [2mm]
\mathrm{subject\ to} &\; X\in\seta_{x},
\end{array}
\end{eqnarray}
where $\seta_{x}$ denotes the set of all possible nonnegative terminal wealth amounts at time $T$ with an initial endowment $x$. Note that $\seta_{x}$ is a convex set, thanks to the linearity of the SDE in \eqref{eq:state-positive}.
The second step is to find an admissible portfolio $\pi^*$ that replicates $X^*$, which will be an optimal portfolio to the original dynamic problem \eqref{eut+}.
As the financial market itself is standard and complete, the second step is standard (see Theorem \ref{Optimalportfolio} below); so we will first focus on the problem \eqref{static1}. 
\par
To solve \eqref{static1}, we need to first characterize the set $\seta_{x}$ in a more tractable form. The following result is well known (see, e.g., Karatazas and Shreve \cite{KS98}).
\begin{lemma}\label{martingalemethod}
We have the following relationship
\begin{align}\label{seta}
\big\{X\in L^{0}_{\BF_{T}}:\BE{\rho X}= x,\; X\geq 0\big\}\subseteq\seta_{x}
\subseteq\big\{X\in L^{0}_{\BF_{T}}:\BE{\rho X}\leq x,\; X\geq 0\big\},
\end{align}
where $\rho$ is called the pricing kernel or the stochastic discount factor, given by
\[ \rho:=\exp\left(-\int_0^{T}\Big(r(s)+\frac{1}{2}\vert \theta(s)\vert^2\Big)\ds-\int_0^{T}\theta(s)^{\tr}\dd W(s)\right). \]
\end{lemma}
The constraint $\BE{\rho X}\leq x$ in \eqref{seta} is called the budget constraint.

According to our standing assumptions, $\theta$ is not identical to zero; so $\rho$ is a genuine random variable (i.e. it is not deterministic).
Also, as
$r$ is bounded and \[ t\mapsto\exp\left(-\frac{1}{2}\int_0^{t} \vert \theta(s)\vert^2 \ds-\int_0^{t}\theta(s)^{\tr}\dd W(s)\right), \quad t \in [0, T],\]
is a martingale (so that its expectation is a constant), we have $\BE{\rho}<\infty$.

Let $\setrho$ be the set of r.v.'s in $L^{0}_{\BF_{T}}$ that are comonotonic with $\rho$ and uniformly distributed on $(0,1)$. Since $\rho\in L^{0}_{\BF_{T}}$,
this set is non-empty and $\rho=Q_{\rho}(U)$ for any $U\in\setrho$ (see, e.g. Xu \cite{X14}). In particular, $\setrho$ is a singleton if and only if $F_{\rho}$ is a continuous function, in which case $\setrho=\{F_{\rho}(\rho)\}$.
This happens if both $r$ and $\theta$ are deterministic processes.

\par
We employ the quantile formulation method to tackle the problem \eqref{static1}, whose key
idea is to change the decision variable from $X$ to its quantile.
Denote by $Q_X$ the quantile function of a given r.v. $X$. The following lemma is critical in reformulating our problem into its quantile formulation.

\begin{lemma}\label{Jproperties}
We have
\begin{align}\label{Jexp}
J_0(X)=\int_0^1 u\big(Q_X(t)+\game(t)\big)\dt.
\end{align}
As a consequence, the functional $J_0$ is law-invariant and strictly increasing in the sense that $J_0(X)=J_0(Y)$ whenever $X\sim Y$ and one of $J_0(X)$ and $J_0(Y)$ is finite, and
$J_0(X_1)> J_0(X_2)$
whenever $X_1\geq X_2$, $\BP{X_1>X_2}>0$ and $J_0(X_2)<\infty$. In particular, \[J_0(X+\ep)> J_0(X)\]
for any constant $\ep>0$ if $J_0(X)<\infty$.
\end{lemma}
\begin{proof}
Recall that two r.v.'s $X$ and $Y$ are called comonotonic if
$$(X(\omega)-X(\omega'))(Y(\omega)-Y(\omega'))\geq 0 \quad\mbox{$\dd\bp\times\dd\bp$ almost surely,}$$ and called anti-comonotonic if $X$ and $-Y$ are comonotonic (see, e.g. Schmeidler \cite{S86} and Yaari \cite{Y87}).
It is well known that comonotonic r.v's have the minimum convex order
(see, e.g., M\"uller \cite{M97}, Dhaene et al. \cite{DDGKV2002a,DDGKV2002b});\footnote{For a pair of r.v.'s $X$ and $Y$, we say that $X$ is smaller than $Y$ in the sense of convex order, if $\BE{g(X)}\leq \BE{g(Y)}$ for all convex functions $g$ such that both expectations exist.} so the infimum in \eqref{J_0def} is obtained if $X$ and $Y^*$ are comonotonic with $Y^*\sim \pr$, in which case $$Q_{X+Y^*}= Q_X+Q_{Y^*}\quad\mbox{a.e.}.$$ Consequently,
\begin{align*}
J_0(X)=\int_0^1 u(Q_{X+Y^*}(t))\dt=\int_0^1 u\big(Q_X(t)+Q_{Y^*}(t)\big)\dt=\int_0^1 u\big(Q_X(t)+\game(t)\big)\dt.
\end{align*}
So \eqref{Jexp} holds.

Next, if $X_1\geq X_2$ and $\BP{X_1>X_2}>0$, then $Q_{X_1}\geq Q_{X_2}$, and $Q_{X_1}(t)> Q_{X_2}(t)$ for some $t\in(0,1)$. By the right-continuity of quantiles, we have $Q_{X_1}>Q_{X_2}$ on some right neighborhood of $t$. Because $u$ is strictly increasing and $J_0(X_2)<\infty$, we have
\[J_0(X_1)=\int_0^1 u\big(Q_{X_1}(t)+\game(t)\big)\dt
> \int_0^1 u\big(Q_{X_2}(t)+\game(t)\big)\dt=J_0(X_2),\]
proving the strictly monotonicity.
\qed
\end{proof}

Let us consider the following relatexed problem
\begin{eqnarray}\label{static2}
\begin{array}{rl}
\sup\limits_{X\in L^{0}_{\BF_{T}} } &\quad J_0(X), \\ [2mm]
\mathrm{subject\ to} &\quad \BE{\rho X}\leq x,\; X\geq 0.
\end{array}
\end{eqnarray}

By the strictly monotonicity of $J_{0}$, we see that any optimal solution $X^{*}$ to the problem \eqref{static2} must satisfy
$\BE{\rho X^{*}}=x$, so $X^{*}\in\seta_{x}$ by \citelem{martingalemethod}.
Hence, any optimal solution to \eqref{static2} also solves the problem \eqref{static1}.
Furthermore, any admissible portfolio that replicates $X^{*}$ will be an optimal portfolio to the problem \eqref{eut+}. Therefore, our problem reduces to solving \eqref{static2}.

Because $J_0$ is law-invariant, henceforth we will define, with a slight abuse of notation:
\[J_0(Q_X):=J_0(X)\equiv \int_0^1 u\big(Q_X(t)+\game(t)\big)\dt.\]

By the expression \eqref{Jexp}, we see that the original problem \eqref{eut+} is not a standard stochastic control problem; hence classical stochastic control theory (e.g. Yong and Zhou \cite{YZ99}) is not applicable directly. To overcome this difficulty, we employ
a refinement of the quantile formulation approach proposed in \cite{X21,X23} to tackle \eqref{eut+}.

Let $\setq$ denote the set of all quantiles generated by nonnegative r.v.'s, that is,
\begin{align*}
\setq :&=\big\{Q: (0,1)\to \R \;\big|\; \text{$Q$ is the quantile for some nonnegative r.v. $X$}\big\}.
\end{align*}
It is easy to see
\begin{align*}
\setq&=\big\{Q: (0,1)\to [0,\infty) \;\big|\; \text{$Q$ is RCI}\big\}.
\end{align*}
Evidently, $\setq$ is a convex set.

\subsection{Well-posedness}
Before tackling our problem \eqref{eut+},
it is natural to ask whether the problem is well-posed; that is, whether its optimal value is finite. An infinite optimal value signals a trivial optimization problem in which trade-offs between competing sub-objectives are not properly modeled.
The following result answers the question of well-posedness completely.

Recall that $ \veum$ and $V_{0}$ are the value functions of the problems \eqref{eut}
and \eqref{eut+}, respectively.

\begin{theorem}\label{wellposedness}
The following statements are equivalent.
\begin{enumerate}
\item $ \veum(x)<\infty$ for some $x>0$.
\item $ \veum(x)<\infty$ for all $x>0$.
\item $V_{0}(x)<\infty$ for all $x>0$.
\item $V_{0}(x)<\infty$ for some $x>0$.
\item There exists a constant $\lam>0$ such that
\[\BE{u\big((u')^{-1}(\lam\rho)\big)}<\infty,\quad \BE{\rho(u')^{-1}(\lam\rho)}<\infty.\]
\end{enumerate}
\end{theorem}
\begin{proof}
The proof consists of the following steps: $1\Longrightarrow 2 \Longrightarrow 3\Longrightarrow 4 \Longrightarrow 1$, $1\Longrightarrow 5$ and
$5\Longrightarrow 1$.

$1\Longrightarrow 2:$ The proof is essentially from \cite{JXZ08}, which is however reproduced here for reader's convenience.
Suppose that $ \veum(x)<\infty$ for some $x>0$.
Fix an arbitrary $y>0$. Then $Y\in\seta_{y}$ if and only if $X=\frac{x}{y}Y\in\seta_{x}$ due to the linearity of the SDE in \eqref{eq:state-positive}.
If $0<y\leq x$, then, because $u$ is increasing,
\begin{align*}
\veum(y)&=\sup_{Y\in\seta_{y}}\;\BE{u(Y)}=\sup_{X\in\seta_{x}}\;\BE{u\Big(\frac{y}{x}X\Big)}\leq \sup_{X\in\seta_{x}}\;\BE{u(X)}=\veum(x)<\infty.
\end{align*}
If $y>x$, then, because $u$ is concave, we have for $k>1$, $a>0$:
\[ \frac{u(ka)-u(a)}{ka-a}\leq \frac{u(a)-u(0)}{a-0},\]
i.e,
\[u(ka)\leq k u(a)-(k-1)u(0).\]
It hence follows
\begin{align*}
\veum(y) &=\sup_{X\in\seta_{x}}\;\BE{u\Big(\frac{y}{x}X\Big)}\\
&\leq \sup_{X\in\seta_{x}}\;\BE{\frac{y}{x} u(X)-\frac{y-x}{x}u(0)} \\
&=\frac{y}{x} \veum(x)-\frac{y-x}{x}u(0)\\
&<\infty.
\end{align*}

$2\Longrightarrow 3:$
Because $u$ is concave, we have for $a,b>0$:
\begin{align*}
\frac{u(a+b)-u(a)}{b}\leq \frac{u(b)-u(0)}{b},
\end{align*}
i.e.,
\begin{align}\label{uineq}
u(a+b)\leq u(a)+u(b)-u(0).
\end{align}
This together with the monotonicity of $u$ yields
\begin{align}
J_0(X)&=\inf\limits_{Y\sim \pr}\;\BE{u(X+Y)}\nn\\
&\leq \inf\limits_{Y\sim \pr}\;\BE{u(X+\max\{Y,0\})}\nn\\
&\leq\inf\limits_{Y\sim \pr}\big(\BE{u(X)+u(\max\{Y,0\})-u(0)}\big)\nn\\
&=\BE{u(X)}+\BE{u(\max\{\pr,0\})-u(0)}.\label{v0leqv}
\end{align}
Maximizing both sides over $X\in\seta_{x}$ yields
\begin{align*}
V_{0}(x)&\leq \veum(x)+\BE{u(\max\{\pr,0\})-u(0)}.
\end{align*}
Thanks to \citeassmp{boundedclaim}, we have
\begin{align*}
\BE{u(\max\{\pr,0\})-u(0)}&=\BE{u(\pr)-u(\min\{\pr,0\})} \leq \BE{u(\pr)}-u(\min\{\game(0),0\})<\infty.
\end{align*}
The implication thus follows.

$3\Longrightarrow 4:$ This is trivial.

$4\Longrightarrow 1:$
Suppose $\game(0)<0$. By the concavity of $u$, for $x> y\geq 1$, we have
\[ \frac{u(x)-u(y)}{x-y}\leq u'(1);\]
so
\[ u(y)\geq u(x)+u'(1)(y-x).\]
This together with the monotonicity of $u$ yields, for any r.v. $X\geq 0$,
\begin{align}
J_0(X)&=\inf\limits_{Y\sim \pr}\;\BE{u(X+Y)}\nn\\
&\geq\inf\limits_{Y\sim \pr}\;\BE{u(X+\einf Y)}
= \BE{u(X+\game(0))}\nn\\
&\geq \BE{u(X+\game(0))\idd{X>1+|\game(0)|}}+\BE{u(\game(0))\idd{0\leq X\leq 1+|\game(0)|}}\nn\\
&\geq \BE{\big(u(X)+u'(1)\game(0)\big)\idd{X>1+|\game(0)|}}- |u(\game(0))|\nn\\
&\geq \BE{u(X)\idd{X>1+|\game(0)|}}-u'(1)|\game(0)|- |u(\game(0))|\nn\\
&=\BE{u(X)}-\BE{u(X)\idd{0\leq X\leq 1+|\game(0)|}}-u'(1)|\game(0)|- |u(\game(0))|\nn\\
&\geq \BE{u(X)}- |u(1+|\game(0)|)|-u'(1)|\game(0)|- |u(\game(0))|.\label{v0geqv}
\end{align}
The last inequality also holds if $\game(0)\geq0$, because in this case
\begin{align*}
J_0(X)&=\inf\limits_{Y\sim \pr}\;\BE{u(X+Y)}\geq
\BE{u(X+\game(0))} \geq \BE{u(X)}.
\end{align*}

Maximizing both sides of \eqref{v0geqv} over $X\in\seta_{x}$ yields
\begin{align*}
V_{0}(x)&\geq \veum(x)-|u(1+|\game(0)|)|-u'(1)|\game(0)|- |u(\game(0))|,
\end{align*}
proving the desired implication.

$1\Longrightarrow 5:$
If $\BE{\rho(u')^{-1}(\lam\rho)}=\infty$ for all $\lam>0$, then by Jin, Xu and Zhou \cite[Theorem 3.1]{JXZ08}, $ \veum(x)=\infty$ for all $x>0$, contradicting Statement $1$.
So there exists some $\lam>0$ such that $$x_{\lam}:=\BE{\rho(u')^{-1}(\lam\rho)}<\infty.$$
Then by the definition of $\veum$, we have
\[\BE{u\big((u')^{-1}(\lam\rho)\big)}\leq \veum(x_\lam).\]
But the right hand side is finite by Statement $2$, which has been proved to be a consequence of Statement $1$. This proves the desired implication.

$5\Longrightarrow 1:$
Let $X_0:=(u')^{-1}(\lam\rho)$ and $x_0:=\BE{\rho X_0}$, then $\BE{u(X_0)}<\infty$ and $x_0<\infty$ by Statement 5.
It thus follows that
\begin{align*}
\veum(x_0)&=\sup_{X\in\seta_{x_0}}\;\BE{u(X)}\\
&\leq \sup_{X\in\seta_{x_0}}\;\BE{u(X)-\lam\rho X}+\lam x_0\\
&\leq \BE{u(X_0)-\lam\rho X_0}+\lam x_0\\
&=\BE{u(X_0)}<\infty,
\end{align*}
where the first equality is due to that $\BE{\rho X}\leq x_0$ for any $X\in\seta_{x_0}$, and
the second equality is due to $\omega$-wise maximization.
The proof is complete.
\qed
\end{proof}

From now on we assume that the value functions of the problems \eqref{eut}
and \eqref{eut+} are both finite, which is technically equivalent to that Statement 5 in Theorem \ref{wellposedness} holds.

\begin{lemma}\label{unboundedV}
The value function $V_{0}$ to the problem \eqref{eut+} is concave and strictly increasing on $(0,\infty)$.
\end{lemma}

\begin{proof}
It follows immediately from the concavity of $u$ that
$J_0$ is concave. The concavity of $V_{0}$ thus comes from the convexity of $\seta_{x}$ and concavity of $J_0$.

Next, we show $$\lim_{x\to\infty} V_{0}(x)=u(\infty).$$
In fact, $X=\frac{x}{\BE{\rho}}$ is a feasible solution to the problem \eqref{static2} with the initial endowment $x>0$; so
\begin{align*}
\liminf_{x\to\infty} V_{0}(x)\geq \lim_{x\to\infty} J_0\Big(\frac{x}{\BE{\rho}}\Big)=\lim_{x\to\infty}\int_0^1 u\Big(\frac{x}{\BE{\rho}}+\game(t)\Big)\dt\geq u(\infty)
\end{align*}
by Fatou's lemma. For any feasible solution $X$ to the problem \eqref{eut+}, since $u$ is increasing,
\begin{align*}
J_0(X)=\int_0^1 u\big(Q_X(t)+\game(t)\big)\dt\leq u(\infty),
\end{align*}
which shows that $V_{0}(x)\leq u(\infty)$. So we conclude $\lim_{x\to\infty} V_{0}(x)=u(\infty)$.

We now show that $V_{0}$ is strictly increasing.
First, it follows from \citelem{Jproperties} that $V_{0}$ is increasing. Suppose $V_{0}$ was not strictly increasing, then it would be a constant on some interval. Because $V_{0}$ is concave and finite at every point, $V_{0}$ must reach its (finite) maximum on that interval. Because $V_{0}$ is increasing and $\lim_{x\to\infty} V_{0}(x)=u(\infty)$, we conclude $ V_{0}(x)=u(\infty)$ for some $x>0$.
This is impossible if $u(\infty)=\infty$ as $V_{0}(x)<\infty$. Hence $V_{0}(x)=u(\infty)<\infty$.
By the monotone convergence theorem, there exists ${t_0}\in (0,1)$ such that
\[\int_{t_0}^1 Q_{\rho}(1-s)\ds>\frac{1}{2}\int_0^1 Q_{\rho}(1-s)\ds=\frac{1}{2}\BE{\rho}.\]
Let $X$ be any feasible solution to the problem \eqref{static2}. Then it follows from the Hardy--Littlewood inequality and the non-negativity and monotonicity of quantiles that
\begin{align*}
x\geq \BE{\rho X} &\geq \int_{0}^1 Q_X(s)Q_{\rho}(1-s)\ds\\
&\geq \int_{t_0}^1 Q_X(s)Q_{\rho}(1-s)\ds\\
&\geq Q_X(t_0)\int_{t_0}^1 Q_{\rho}(1-s)\ds\geq \frac{1}{2}Q_X(t_0)\BE{\rho},
\end{align*}
which yields
\begin{align*}
J_0(X)&=\int_0^1 u\Big(Q_X(s)+\game(s)\Big)\ds\\
&\leq\int_0^{t_0} u\Big(\frac{2x}{\BE{\rho}}+\game(s)\Big)\ds
+\int_{t_0}^1 u(\infty)\ds.
\end{align*}
Maximizing both sides over $X\in\seta_{x}$ leads to
\begin{align*}
V_{0}(x)&\leq \int_0^{t_0} u\Big(\frac{2x}{\BE{\rho}}+\game(s)\Big)\ds
+\int_{t_0}^1 u(\infty)\ds \\
&<\int_0^{t_0}u(\infty)\ds+\int_{t_0}^1 u(\infty)\ds \\
&=u(\infty),
\end{align*}
where the last inequality is due to the strictly monotonicity of $u$ and the fact that $u(\infty)<\infty$. This contradicts that $V_{0}(x)=u(\infty)$. The proof is thus complete.
\qed
\end{proof}

By \citelem{Jproperties}, the functional $J_0$ is law-invariant and strictly increasing, so we can now apply \cite[Theorem 9]{X14} to obtain the following result.
\begin{proposition}\label{quantile}
A random variable $X^*\in L^{0}_{\BF_{T}}$ is an optimal solution to the problem \eqref{static1} if and only if
it can be expressed as
\[X^*=\barq(1-U)\]
where $U\in\setrho$ and $\barq$ is an optimal solution to the following quantile optimization problem
\begin{align}\label{q1}
\sup\limits_{Q\in\setq} &\;\; J_{0}(Q) \\
\mathrm{subject\ to} &\;\; \int_0^1 Q(t)Q_{\rho}(1-t)\dt=x.\nn
\end{align}
Moreover, the optimal value of the problem \eqref{q1} is equal to $V_0(x)$, the optimal value of the problem \eqref{eut+}.
\end{proposition}
By virtue of this result, solving our original stochastic control problem \eqref{eut+} reduces to solving the quantile optimization problem \eqref{q1}.

\section{Solution}\label{secqfsol}

We now solve \eqref{q1}, which is a constrained concave optimization (or calculus of variations) problem. The first step is to remove the constraint using the Lagrange method.

\begin{lemma}\label{Lagrangemethod}
A quantile $\barq$ is an optimal solution to \eqref{q1} if and only if there exists a constant $\lam>0$ such that $\barq$ satisfies
\begin{align}\label{q2a}
\int_0^1 \barq(t)Q_{\rho}(1-t)\dt=x,
\end{align}
and solves the following optimization problem
\begin{align}\label{q2}
V_{\lam}=\sup\limits_{Q\in\setq} &\;\; J_{\lam}(Q),
\end{align}
where
\begin{align*}
J_{\lam}(Q):=\int_0^1\big[u\big(Q(t)+\game(t)\big) -\lam Q(t)Q_{\rho}(1-t)\big]\dt.
\end{align*}
Moreover,
\[V_0(x)=V_{\lam}+\lam x.\]
\end{lemma}
\begin{proof}
$\Longrightarrow:$ Because $V_{0}$ is concave, for each given $x>0$, we have
\begin{align}\label{lag1}
V_{0}(y)-\lam y\leq V_{0}(x)-\lam x,\quad y>0,
\end{align}
where
\[\lam=\liminf_{\Delta\to 0+}\frac{V_{0}(x+\Delta)-V_{0}(x)}{\Delta}\geq 0,\]
due to the monotonicity of $V_0$.
By the concavity and finiteness of $V_{0}$,
\[\lam=\liminf_{\Delta\to 0+}\frac{V_{0}(x+\Delta)-V_{0}(x)}{\Delta}\leq \frac{V_{0}(x)-V_{0}(x/2)}{x-x/2}<\infty.\]
If $\lam=0$, then \eqref{lag1} would indicate that $x$ is a global maximizer of $V_{0}$, contradicting the strictly monotonicity of $V_{0}$ from \citelem{unboundedV}.
So we conclude $0<\lam<\infty$.

Suppose $\barq$ is an optimal solution to \eqref{q1}. Then it clearly satisfies \eqref{q2a}.
Also,
\begin{align}\label{eq1}
V_{0}(x)=J_{0}(\barq)=J_{\lam}(\barq)+\lam\int_0^1\barq(t)Q_{\rho}(1-t)\dt=J_{\lam}(\barq)+\lam x.
\end{align}
For any $Q\in\setq$, let
\[Q_n=\min\{Q,n\},\]
and
\[ y_n=\int_0^1Q_n(t)Q_{\rho}(1-t)\dt\leq \int_0^1nQ_{\rho}(1-t)\dt=n\BE{\rho}<\infty. \]
Then, by definition,
\begin{align*}
J_{\lam}(Q_n)&=J_{0}(Q_n)-\lam \int_0^1Q_n(t)Q_{\rho}(1-t)\dt\\
&=J_{0}(Q_n)-\lam y_n\\
&\leq V_{0}(y_n)-\lam y_n\leq V_{0}(x)-\lam x=J_{\lam}(\barq)
\end{align*}
in view of \eqref{lag1} and \eqref{eq1}.
It hence follows
\begin{align*}
J_{\lam}(\barq)\geq J_{\lam}(Q_n)\geq \int_0^1 u\big(Q_n(t)+\game(t)\big)\dt -\lam \int_0^1 Q(t)Q_{\rho}(1-t)\dt.
\end{align*}
Because $$u\big(Q_n(t)+\game(t)\big)\geq u\big(\game(t)\big)\geq u\big(\game(0)\big),$$
Fatou's lemma yields
\begin{align*}
J_{\lam}(\barq) &\geq \liminf_{n\to\infty} \int_0^1 u\big(Q_n(t)+\game(t)\big)\dt -\lam \int_0^1 Q(t)Q_{\rho}(1-t)\dt\\
&\geq\int_0^1 u\big(Q(t)+\game(t)\big)\dt -\lam \int_0^1 Q(t)Q_{\rho}(1-t)\dt
=J_{\lam}(Q).
\end{align*}
This indicates that $\barq$ is an optimal solution to \eqref{q2}.

$\Longleftarrow:$ Suppose there exists $\lam>0$ such that $\barq$ satisfies \eqref{q2a} and is an optimal solution to \eqref{q2}.
Then, for any feasible solution $Q\in\setq$ to \eqref{q1},
\begin{align*}
J_{0}(Q)=J_{\lam}(Q)+\lam\int_0^1 Q(t)Q_{\rho}(1-t)\dt&=J_{\lam}(Q)+\lam x\\
&\leq J_{\lam}(\barq)+\lam x=J_{\lam}(\barq)+\lam\int_0^1\barq(t)Q_{\rho}(1-t)\dt=J_{0}(\barq).
\end{align*}
Because $\barq$ is a feasible solution to \eqref{q1}, the above inequality shows that it is optimal to \eqref{q1}.
\qed
\end{proof}

Our problem now reduces to first solving the unconstrained optimization \eqref{q2}, and then finding $\lam>0$ such that the resulting optimal solution satisfies \eqref{q2a}.
The quantile optimization \eqref{q2} is a concave optimization problem, which can be tackled by calculus of variations following Xu \cite{X21,X23}.

\subsection{Well-posedness of \eqref{q2}}

Before solving \eqref{q2}, we need to first address its well-posedness.

\begin{lemma}\label{finiteoptimal1}
The problem \eqref{q2} is well-posed if and only if
\begin{align} \label{finitecon1}
\BE{u\big((u')^{-1}(\lam\rho)\big)-\lam \rho (u')^{-1}(\lam\rho)}<\infty.
\end{align}
\end{lemma}
\begin{proof}
By the estimates \eqref{v0leqv} and \eqref{v0geqv}, we have
\begin{align*}
-C_1\leq J_0(X)-\BE{u(X)} \leq C_2,\;\; \forall\; X\in L^{0}_{\BF_{T}},
\end{align*}
where
\begin{align*}
C_1&:=|u(1+|\game(0)|)|+u'(1)|\game(0)|+|u(\game(0))|,\\
C_2&:=\BE{u(\max\{\pr,0\})}-u(0).
\end{align*}
Hence,
\begin{align*}
-C_1\leq J_{\lam}(Q) - \int_0^1\big[u\big(Q(t)\big) -\lam Q(t)Q_{\rho}(1-t)\big]\dt\leq C_2,\;\; \forall\; Q\in \setq.
\end{align*}
Maximizing both sides over $Q\in \setq$ yields
\begin{align}
-C_1\leq V_{\lam} - \sup_{Q\in \setq}\;\int_0^1\big[u\big(Q(t)\big) -\lam Q(t)Q_{\rho}(1-t)\big]\dt\leq C_2. \label{diff1}
\end{align}
Let \[Q_0(t):=\lim_{s\to t+}(u')^{-1}(\lam Q_{\rho}(1-s)).\] Then $Q_0\in\setq$. Applying the first-order condition to the maximum of the function $u\big(x\big) -\lam xQ_{\rho}(1-t)$ in $x$, we have
\begin{align*}
u\big(Q(t)\big) -\lam Q(t)Q_{\rho}(1-t) &\leq u\big((u')^{-1}(\lam Q_{\rho}(1-t))\big) -\lam (u')^{-1}(\lam Q_{\rho}(1-t)) Q_{\rho}(1-t).
\end{align*}
The right hand side is, for a.e. $t\in(0,1)$, equal to
\begin{align*}
u\big(Q_0(t)\big) -\lam Q_0(t)Q_{\rho}(1-t).
\end{align*}
Hence, we conclude that
\begin{align*}
\sup_{Q\in \setq}\;\int_0^1\big[u\big(Q(t)\big) -\lam Q(t)Q_{\rho}(1-t)\big]\dt&=\int_0^1\big[u\big(Q_0(t)\big) -\lam Q_0(t)Q_{\rho}(1-t)\big]\dt.
\end{align*}
For $U\in\setrho$, we have $Q_{\rho}(U)=\rho$ and
\[ Q_0(1-U)=(u')^{-1}\big(\lam Q_{\rho}(U)\big)=(u')^{-1}(\lam\rho), \quad\mbox{a.s.},\]
so
\begin{align} \label{diff2}
\sup_{Q\in \setq}\;\int_0^1\big[u\big(Q(t)\big) -\lam Q(t)Q_{\rho}(1-t)\big]\dt
&=\BE{u\big((u')^{-1}(\lam\rho)\big)-\lam \rho (u')^{-1}(\lam\rho)}.
\end{align}
This together with \eqref{diff1} implies that the optimal value of the problem \eqref{q2} is finite if and only if the condition \eqref{finitecon1} holds. The proof is complete.
\qed
\end{proof}
\begin{remark}
Because
\[\frac{\partial}{\partial\lam}\BE{u\big((u')^{-1}(\lam\rho)\big)-\lam \rho (u')^{-1}(\lam\rho)}=-\BE{ \rho (u')^{-1}(\lam\rho)}<0,\]
the condition \eqref{finitecon1} holds if and only if either $\lam>\llam_1$ or $\lam\geq \llam_1$,
where
\[\llam_1=\inf\Big\{\lam>0:\BE{u\big((u')^{-1}(\lam\rho)\big)-\lam \rho (u')^{-1}(\lam\rho)}<\infty\Big\}.\]
\end{remark}

\subsection{Characterization of solution}
We now investigate the problem \eqref{q2}. In view of \citelem{finiteoptimal1}, we only need to study the case under the well-posedness condition \eqref{finitecon1},
which is henceforth assumed.

Because $u$ is strictly concave, \eqref{q2} admits at most one optimal solution. The following result characterizes the unique optimal solution, if it exists.
\begin{lemma}[Optimality condition]\label{op1}
A quintile $\barq\in\setq$ is the unique optimal solution to \eqref{q2} if and only if it satisfies
\begin{align} \label{opc}
\int_{0}^{1}\big[u'\big(\barq(t)+\game(t)\big)-\lam Q_{\rho}(1-t)\big]\big(Q(t)-\barq(t)\big)\dt\leq 0, \quad
\forall\; Q\in\setq.
\end{align}
\end{lemma}

\begin{proof}
Suppose $\barq$ is the unique optimal solution to \eqref{q2}. For any $Q\in\setq$, $\ep\in(0,1)$, define
\[Q_{\ep}(t)=\barq(t)+\ep \big(Q(t)-\barq(t)\big),\quad t\in(0,1).\]
Then $Q_{\ep}\in\setq$. Recalling that $\barq$ is optimal and applying Fatou's lemma, we get
\begin{align*}
0 &\geq \liminf_{\ep\to 0+}\frac{1}{\ep}\bigg[\int_{0}^{1}u(Q_{\ep}(t)+\game(t)) -\lam Q_{\ep}(t)Q_{\rho}(1-t)\dt\\
&\qquad\qquad\qquad-\int_{0}^{1}u(\barq(t)+\game(t)) -\lam \barq(t)Q_{\rho}(1-t)\dt\bigg]\\
&\geq \int_{0}^{1}\big[u'\big(\barq(t)+\game(t)\big)-\lam Q_{\rho}(1-t)\big]\big(Q(t)-\barq(t)\big)\dt,
\end{align*}
leading to \eqref{opc}.
\par
Conversely, suppose $\barq\in\setq$ satisfies \eqref{opc}.
Because $u$ is concave, we have the elementary inequality $u(y)-u(x)\leq u'(x)(y-x)$ for any $x$, $y\in\R$. It follows that
\begin{multline*}
\qquad\; u\big(Q(t)+\game(t)\big) -\lam Q(t)Q_{\rho}(1-t)-\big[u(\barq(t)+\game(t)) -\lam \barq(t)Q_{\rho}(1-t)\big]\\
\leq \big[u'\big(\barq(t)+\game(t)\big)-\lam Q_{\rho}(1-t)\big]\big(Q(t)-\barq(t)\big),\quad \forall\; Q\in\setq.\qquad\qquad
\end{multline*}
Integrating both sides in above and using \eqref{opc}, we conclude $\barq$ is optimal to \eqref{q2}.
\qed
\end{proof}

Before proceeding further, we impose the following technical assumption in the rest of this paper.
\begin{assumption}\label{quantiles1}
The quantiles of $\pr$ and ${\rho}$ satisfy
\[\lim_{t\to 0+}\frac{u'(\game(t))}{Q_{\rho}(1-t)}=0.\]
\end{assumption}
This assumption is satisfied if $\rho$ is unbounded and $u'(\game(0))<\infty$.
Moreover, this assumption together with $\BE{\rho}<\infty$ implies
\begin{align} \label{integral1}
\BE{u'(\pr)}=\int_0^1 u'(\game(t))\dt<\infty.
\end{align}

\begin{corollary}\label{initial}
The optimal solution $\barq$ to the problem \eqref{q2} must satisfy $\barq(0)=0$.
\end{corollary}
\begin{proof}
Because $\barq\geq 0$ and $u$ is concave, \citeassmp{quantiles1} implies
\[u'\big(\barq(t)+\game(t)\big)-\lam Q_{\rho}(1-t)\leq u'\big(\game(t)\big)-\lam Q_{\rho}(1-t)<0, \; \mbox{ for $t\in(0, \ep)$,}\] with $\ep>0$ being sufficiently small.
Suppose $\barq(0)>0$, then we may assume $\barq(t)>0$ for $t\in(0,\ep)$ by the right-continuity of $\barq$.
Now let $Q(t)=\barq(t)\idd{t\geq \ep}$. Then $Q\in\setq$ and
\begin{align*}
\int_{0}^{1}\big[u'\big(\barq(t)+\game(t)\big)-\lam Q_{\rho}(1-t)\big]\big(Q(t)-\barq(t)\big)\dt> 0,
\end{align*}
which contradicts \eqref{opc}. The proof is complete.
\qed
\end{proof}

\par
It is hard to use the condition \eqref{opc} to find the optimal solution to the problem \eqref{q2} because one would have to compare $\barq$ with all the other quantiles in $\setq$, a task as difficult as solving \eqref{q2}.
Our next step is to find an equivalent condition to \eqref{opc} that can be easily verified and utilized. To this end, let
\begin{align*}
H(t):=\int_t^1 \big[\lam Q_{\rho}(1-s) -u'\big(\barq(s)+\game(s)\big)\big]\ds.
\end{align*}
Using $\BE{\rho}<\infty$, \eqref{integral1} and the monotonicity of $u'$, one can easily show that $H$ is a continuous function on $[0,1]$.
In terms of $H$, the inequality in \eqref{opc} now reads
\begin{align*}
\int_{0}^{1}H'(t) \big(Q(t)-\barq(t)\big)\dt\leq 0.
\end{align*}
By taking $Q=2\barq$ and $Q=\frac{1}{2}\barq$ in above, we obtain the following two conditions
\begin{align} \label{opc1-2b}
\int_{0}^{1}H'(t) \barq(t)\dt=0,
\end{align}
and
\begin{align} \label{opc1-2c}
\int_{0}^{1}H'(t) Q(t)\dt\leq 0.
\end{align}

Suppose $H(a)<0$ for some $a\in [0,1)$. Let $Q(t)= \idd{t\in [a, 1)}$, then $Q\in\setq$ and
\begin{align*}
\int_{0}^{1}H'(t) Q(t)\dt=-H(a)>0,
\end{align*}
contradicting \eqref{opc1-2c}. Hence, $H\geq 0$ on $[0,1]$ by continuity.

Suppose $H>0$ on an interval $(a,b]\subset (0,1)$ and $\barq(b)>\barq(a)$. Then since $H\geq 0$, $H(1)=0$ and by integration by parts, we have
\begin{align*}
\int_{0}^{1}H'(t) \barq(t)\dt=H(0)\barq(0)+\int_{(0,1)} H(t)\dd \barq(t)
\geq \int_{(a,b]} H(t)\dd \barq(t)>0,
\end{align*}
contradicting \eqref{opc1-2b}. This implies that $\barq$ is constant on every subinterval of $\big\{t\in(0,1)\;\big|\;H(t)>0\big\}$, and consequently
$\barq'=0$ on $\big\{t\in(0,1)\;\big|\;H(t)>0\big\}$.

Together with \citecoro{initial}, we conclude that
if $\barq$ is the optimal solution to the problem \eqref{q2}, then
\begin{align} \label{opc2}
\begin{cases}
\dsp \min\Big\{\barq'(t), \; \int_t^1 \big[\lam Q_{\rho}(1-s)-u'\big(\barq(s)+\game(s)\big)\big]\ds\Big\}=0,\; \mbox{ for $t\in(0,1)$;}\medskip\\
\barq(0)=0.
\end{cases}
\end{align}
This is a so-called ordinary integro--differential equation (OIDE) with boundary condition for $\barq$.
It is easy to check that \eqref{opc2} implies \eqref{opc1-2b} and \eqref{opc1-2c}, which is equivalent to \eqref{opc}.
Therefore, it follows from \citelem{op1} that the solution to \eqref{opc2} is indeed the optimal solution to \eqref{q2}.
Clearly, in comparison to \eqref{opc}, the condition \eqref{opc2} is easier to verify and use because it only depends on $\barq$ itself.

Summarizing the preceding analysis we present the following main result.
\begin{theorem}[Characterization of optimal solution I]\label{main1}
A function $\barq$ is the unique optimal solution to the problem \eqref{q2} if and only if it is a RCI function on $(0,1)$ that satisfies the OIDE \eqref{opc2}.
\end{theorem}
\begin{corollary}
If $Q_{\rho}$ is continuous, then so is $\barq$, the optimal solution to the problem \eqref{q2}.
Moreover, we have either $\barq'(t)=0$ or
\begin{align} \label{barq1}
\barq(t)=(u')^{-1}(\lam Q_{\rho}(1-t))-\game(t)
\end{align}
for every $t\in(0,1)$. In particular, the latter case can only happen if $\game$ is continuous at $t$.
\end{corollary}

\begin{proof}
We have shown earlier that on each subinterval of $\{t\in(0,1)\;|\;H(t)>0\}$, $\barq$ is a constant, and hence continuous there. Because quantiles are right-continuous, $\barq$ is continuous at $0$ as well.

Now suppose $H(t)=0$ at some $t\in(0,1)$.
Because quantiles are RCI, $u'$ and $Q_{\rho}$ are continuous, we have
\begin{align*}
\lim_{\Delta\to 0+}\frac{H(t+\Delta)-H(t)}{\Delta} &= u'\big(\barq(t)+\game(t)\big)-\lam Q_{\rho}(1-t),\\
\lim_{\Delta\to 0+}\frac{H(t)-H(t-\Delta)}{\Delta}&=u'\big(\barq(t-)+\game(t-)\big)-\lam Q_{\rho}(1-t).
\end{align*}
Notice that $t$ is a global minimizer of $H$, so the first limit is nonnegative and the second one is nonpositive, giving
\[u'\big(\barq(t)+\game(t)\big)-\lam Q_{\rho}(1-t)\geq u'\big(\barq(t-)+\game(t-)\big)-\lam Q_{\rho}(1-t).\]
Because $u'$ is strictly decreasing, it follows that
\[\barq(t)+\game(t)\leq \barq(t-)+\game(t-).\]
But quantiles are increasing, the above should be an equation, which let us conclude $\barq$ and $\game$ are both continuous at $t$. This further shows that $H$ is differentiable at $t$, and
\[0=H'(t)=u'\big(\barq(t)+\game(t)\big)-\lam Q_{\rho}(1-t).\]
Consequently,
\begin{align*}
\barq(t)=(u')^{-1}(\lam Q_{\rho}(1-t))-\game(t),
\end{align*}
completing the proof.
\qed
\end{proof}

As much as OIDEs such as \eqref{opc2} are more accessible than the general condition \eqref{opc}, they are still hard to analyze in general. Next, we further simplify the condition by turning \eqref{opc2} into an ODE.

Denote
\begin{align}\label{opldef1}
\opl(t)=\int_{t}^{1}u'\big(\barq(s)+\game(s)\big)\ds.
\end{align}
Then $\opl\in \setcb$ and
\begin{align*}
\barq(t)=(u')^{-1}(-\opl'(t))-\game(t), \quad H(t)=-\opl(t)+\lam\int_t^1 Q_{\rho}(1-s)\ds.
\end{align*}
Noting that $\opl''$ exists almost everywhere because $\opl\in\setcb$, we can rewrite the OIDE \eqref{opc2} as
\begin{align}\label{vi001-1}
\min\left\{\frac{-\opl''(t)}{u''\big((u')^{-1}(-\opl'(t))\big)}-\game'(t),\; -\opl(t)+\lam\int_t^1 Q_{\rho}(1-s)\ds\right\}=0, \mbox{\quad a.e. } t\in(0,1).
\end{align}
Using the following simple fact that $\min\{a, \;b\}=0$ if and only if $\min\{ak, b\ell\}=0$ for any $k, \ell>0$, together with the boundary conditions, we end up with the following equation for $\opl$:
\begin{align} \label{vi002}
\begin{cases}
\min\Big\{\opl''(t)+\game'(t)u''\big((u')^{-1}(-\opl'(t))\big),\; -\opl(t)+\lam\int_t^1 Q_{\rho}(1-s)\ds\Big\}=0, \mbox{\quad a.e. } t\in(0,1),\\
\opl(1-)=0, \quad \opl'(0+)=-u'(\game(0)).
\end{cases}
\end{align}
In \eqref{vi002}, if $\opl'$ or $\game$ is not differentiable at some point $t\in (0,1)$, we convent
\[\opl''(t)+\game'(t)u''\big((u')^{-1}(-\opl'(t))\big)=+\infty,\]
in which case the equation reduces to \[\opl(t)=\lam\int_t^1 Q_{\rho}(1-s)\ds,
\]
or $\opl'(t)=-\lam Q_{\rho}(1-t)$
whenever $Q_{\rho}$ is continuous at $t$.
Also note that it can happen that $u'(\game(0))=+\infty$.

Equation \eqref{vi002} is a system of variational inequalities of ODEs, or a
single--obstacle second-order ODE with mixed boundary conditions.
The connection between \eqref{vi002} and the optimal solution to the problem \eqref{q2} is given as follows.
\begin{theorem}[Characterization of optimal solution II]\label{main1}
We have the following assertions.
\begin{enumerate}[(1).]
\item
If $\barq$ is the unique optimal solution to the problem \eqref{q2}, then
\begin{align}\label{opldef3}
\opl(s)=\int_{s}^{1}u'\big(\barq(t)+\game(t)\big)\dt
\end{align}
is a solution to \eqref{vi002} in $\setcb$.

\item
If $\opl$ is a solution to \eqref{vi002} in $\setcb$, then
\begin{align*}
\barq(t)=(u')^{-1}(-\opl'(t))-\game(t)
\end{align*}
is the unique optimal solution to the problem \eqref{q2}.
\end{enumerate}
As a consequence, \eqref{vi002} admits at most one solution in $\setcb$.
\end{theorem}
\begin{proof}
\begin{enumerate}[(1).]
\item Because $\barq$ solves \eqref{q2}, we have
\eqref{vi001-1}. Consequently, we get the first equation in \eqref{vi002}.
The first boundary condition $\opl(1-)=0$ is evident.
It follows from \citecoro{initial} that $\barq(0)=0$, which implies the second boundary condition $\opl'(0+)=-u'(\game(0))$.
Moreover, as quantiles are RCI and $u'$ is continuously decreasing, it is easy to verify that $\opl\in\setcb$.
\item
Because $\opl$ solves \eqref{vi002} in $\setcb$, it satisfies \eqref{vi001-1}. Hence
\begin{align*}
\barq(t)=(u')^{-1}(-\opl'(t))-\game(t)
\end{align*}
satisfies \eqref{opc2}, which in turn is equivalent to f\eqref{opc}.
As $\opl\in \setcb$ and $\game$ is RCI, $\barq$ is right-continuous. By \eqref{opc2}, $\barq$ is also increasing; so $\barq$ is a quantile.
It follows then from \citelem{op1} that $\barq$ is the unique optimal solution to \eqref{q2}.
\end{enumerate}
\qed
\end{proof}

The following theorem provides an optimal portfolio to the problem \eqref{eut+}.
Its proof is standard (see \cite{KS98}); but we reproduce it here for the readers' convenience.

\begin{theorem}[Optimal portfolio]\label{Optimalportfolio}
Let $\lambda>0$ and $\opl$ be a solution to \eqref{vi002} in $\setcb$ such that
\begin{align}\label{budget}
\int_0^1 \barq(t)Q_{\rho}(1-t)\dt=x
\end{align}
where
\begin{align*}
\barq(t)=(u')^{-1}(-\opl'(t))-\game(t), \quad t\in[0,T],
\end{align*}
and $U$ be arbitrary chosen from $\setrho$.
Then an optimal portfolio to the problem \eqref{eut+} is given in a feedback form as follows:
\begin{align} \label{optimalportfoliopi}
\pi^{*}(t)
&=\exp\left(\int_0^{t}\Big(r(s)+\frac{1}{2}\vert \theta(s)\vert^2\Big)\ds
+\int_0^{t}\theta(s)^{\tr}\dd W(s)\right) \big(\sigma(t)^{\tr}\big)^{-1}Z(t)+\big(\sigma(t)^{\tr}\big)^{-1}\theta(t)X^{\pi^{*}}(t),
\end{align}
whereas the corresponding optimal wealth process is given by
\[X^{\pi^{*}}(t)=\exp\left(\int_0^{t}\Big(r(s)+\frac{1}{2}\vert \theta(s)\vert^2\Big)\ds+\int_0^{t}\theta(s)^{\tr}\dd W(s)\right)\BE{Q_{\rho}(U)\barq(1-U)\;\Big|\;\BF_t},\quad t\in[0,T],\]
where $Z$ is an $\{{\cal F}_t\}_{t\geq0}$-progressively measurable process such that
\[\int_{0}^{T}|Z(s)|^{2}\ds<\infty,\;\;\mbox{a.s.}\]
and
\begin{align*}
\BE{Q_{\rho}(U)\barq(1-U)\;\Big|\;\BF_t}
=x+\int_{0}^{t}Z(s)^{\tr}\dd W(s),\quad t\in[0,T].
\end{align*}
\end{theorem}

\begin{proof}
\citethm{main1} yields that $\barq$ is an optimal solution to the problem \eqref{q2}.
Thanks to \citelem{Lagrangemethod} and \eqref{budget}, we conclude $\barq$ also solves \eqref{q1}.

Let $X^{*}=\barq(1-U)$. Since $\barq\in\setq$, it yields $X^{*}\geq 0$.
Since $U\in L^{0}_{\BF_{T}}$, we have $X^{*}\in L^{0}_{\BF_{T}}$.
It follows from \eqref{budget} and $\rho=Q_{\rho}(U)$ that
\begin{align*}
\BE{\rho X^*}=\BE{Q_{\rho}(U) \barq(1-U)}=\int_0^1 Q_{\rho}(t)\barq(1-t)\dt=x;
\end{align*}
so $X^*\in\seta_{x}$.
By \citeprop{quantile}, we see that $X^*$ is an optimal solution to the problem \eqref{static2} and hence solves \eqref{static1}.
Because our market is standard and complete, we can follow \cite[Theorem 6.6, page 24]{KS98} to conclude that $X^*$ can be replicated by an admissible strategy.
Here we give the major steps of the proof.

Define an $\{{\cal F}_t\}_{t\geq0}$-progressively measurable process $\varrho$ by
\begin{align}\label{varrho}
\varrho(t)=\exp\left(-\int_0^{t}\Big(r(s)+\frac{1}{2}\vert \theta(s)\vert^2\Big)\ds-\int_0^{t}\theta(s)^{\tr}\dd W(s)\right),\quad t\in[0,T].
\end{align}
Notice $\varrho(T)=\rho=Q_{\rho}(U)$ and
\[\dd \varrho(t)^{-1}=\varrho(t)^{-1}\Big(\big(r(t)+|\theta(t)|^{2}\big)\dt+\theta(t)^{\tr}\dd W(t)\Big).\]
Let
\begin{align*}
Y(t)=\varrho(t)^{-1} \BE{\varrho(T) X^{*}\;\Big|\;\BF_t}=\varrho(t)^{-1} \BE{Q_{\rho}(U)\barq(1-U)\;\Big|\;\BF_t},\quad t\in[0,T].
\end{align*}
Then $Y$ is an $\{{\cal F}_t\}_{t\geq0}$-progressively measurable process.
Clear, $\varrho Y$ is a martingale and $Y(T)=X^{*}$.
Because our market is standard,
by \cite[Lemma 6.7, page 24]{KS98},
there exists an $\{{\cal F}_t\}_{t\geq0}$-progressively measurable process $Z$ such that
\[\int_{0}^{T}|Z(s)|^{2}\ds<\infty, \;\;\mbox{a.s.}\]
and
\begin{align*}%\label{martingale2}
\varrho(t)Y(t)=\BE{\varrho(T)Y(T)}+\int_{0}^{t}Z(s)^{\tr}\dd W(s)
=x+\int_{0}^{t}Z(s)^{\tr}\dd W(s),\quad t\in[0,T].
\end{align*}
An application of It\^{o}'s lemma to the product of $\varrho(t)^{-1}$ and $\varrho(t)Y(t)$ leads to
\begin{equation*}%\label{eq:Y}
\dd Y(t)=\Big[\big(r(t)+|\theta(t)|^{2}\big)Y(t)+\varrho(t)^{-1}Z(t)^{\tr}\theta(t)\Big]\dt+
\Big(\varrho(t)^{-1}Z(t)+Y(t)\theta(t)\Big)^{\tr}\dd W(t), \quad t \in [0, T].
\end{equation*}
We set
\begin{align*}
\pi^{*}(t)=\big(\sigma(t)^{\tr}\big)^{-1}\Big(\varrho(t)^{-1}Z(t)+Y(t)\theta(t)\Big), \quad t \in [0, T].
\end{align*}
Then it is an $\{{\cal F}_t\}_{t\geq0}$-progressively measurable process such that
\begin{align}\label{eq:Y2}
\dd Y(t)=\Big[r(t)Y(t)+\pi^{*}(t)^{\tr}\sigma(t)\theta(t)\Big]\dt+
\pi^{*}(t)\sigma(t)^{\tr}\dd W(t), \quad t \in [0, T].
\end{align}
Comparing it to the wealth process \eqref{eq:state-positive} and noticing $Y(0)=x$,
we conclude that
\[Y(t)=X^{\pi^{*}}(t), \quad t \in [0, T].\]
In particular, $X^{\pi^{*}}(T)=Y(T)=X^{*}$, which means that $\pi^{*}$ is a portfolio that replicates $X^{*}$ with the initial endowment $x$.
Because $\varrho^{-1}$ and $Y$ are continuous processes (so that they are $\omega$-wisely bounded), and both $\theta$ and $Z$ are square integrable processes, we conclude
\[\int_{0}^{T}|\sigma(t)^{\tr}\pi^{*}(s)|^{2}\ds<\infty, \;\;\mbox{a.s.}.\]
Thus, $\pi^{*}$ is an admissible portfolio.
Since $\pi^{*}$ replicates an optimal solution $X^*$ to the problem \eqref{static1}, we conclude that it is an optimal portfolio to the problem \eqref{eut+}.
\qed
\end{proof}

Note that if the set $\setrho$ is not a singleton, then the choice of $U$ is not unique and so the optimal portfolio to the problem \eqref{eut+} may not be unique.

In the above result we require the existence of the Lagrange multiplier $\lambda>0$.
Such an assumption needs to be checked for a more concrete problem. The assumption may not hold even in the classical case without intractable claims; refer to \cite{JXZ08} for a comprehensive discussion in the classical case.

\section{A Special Case of Exponential Utility}\label{expsp}

We have reduced the original problem to solving the equation \eqref{vi002}. Such an equation can be in general solved numerically, but more concrete solutions are possible for special cases. In this section we study such a case where the utility function $u(x)=-e^{-\alpha x}$ for some $\alpha>0$ and $\game$ is continuous. Then \eqref{vi002} becomes
\begin{align} \label{vi003}
\begin{cases}
\min\Big\{\opl''(t)+\alpha\game'(t)\opl'(t),\; -\opl(t)+\lam\int_t^1 Q_{\rho}(1-s)\ds\Big\}=0, \quad \mbox{a.e.}\; t\in(0,1),\\
\opl(1-)=0, \quad \opl'(0+)=-\alpha e^{-\alpha \game(0)}.
\end{cases}
\end{align}
Note the second-order differential operator is now linear, which enables us to give a more explicit solution.

For any constant $c\geq 0$, consider the following integral equation for $\chi$:
\[ \chi(t)=c\int_t^1 e^{\alpha \game(\chi(s))}\ds,\quad t\in [0,1].\]
Clearly, any solution $\chi(t)$ is continuous in $t$.
Let $\overline{\chi}=\chi/c$. Then
\[ \overline{\chi}(t)=\int_t^1 e^{\alpha \game(c\overline{\chi}(s))}\ds,\quad t\in [0,1].\]
It follows from the monotone convergence theorem that $\overline{\chi}(t)$ is increasing in $c$; hence $\chi(t)$ is strictly increasing in $c$ for every $t\in[0,1)$.
Since $\game$ is continuous, $\overline{\chi}(t)$ and thus ${\chi}(t)$ are continuous in $c$ for any fixed $t\in[0,1)$.
If $c=0$, then $\chi(0)=0$. If $c>e^{-\alpha \game(0)}$, then, since $\chi\geq 0$ and
and $\game$ is increasing,
\[ \chi(0)=c\int_0^1 e^{\alpha \game(\chi(s))}\ds\geq c\int_0^1 e^{\alpha \game(0)}\ds>1.\]
So there exists a constant $c>0$ such that $\chi(0)=1$, leading to the following equations
\[ \chi'(t)=-ce^{\alpha \game(\chi(t))}, \quad \chi(0)=1, \quad \chi(1)=0.\]
Then
\[\chi''(t)=-ce^{\alpha \game(\chi(t))}\alpha\game'(\chi(t))\chi'(t)=\alpha\game'(\chi(t))(\chi'(t))^2,\quad \mbox{a.e.}\; t\in(0,1).\]
Let $\opll(t)=\opl(\chi(t))$, then $\opll'(t)=\opl'(\chi(t))\chi'(t)$ and
\[\opll''(t)=\opl''(\chi(t))(\chi'(t))^2+\opl'(\chi(t))\chi''(t)
=\big(\opl''(\chi(t))+\alpha\game'(\chi(t))\opl'(\chi(t))\big)(\chi'(t))^2.\]
Replacing $t$ by $\chi(t)$ in \eqref{vi003} and using the above facts, we get
\begin{align} \label{vi004}
\begin{cases}
\min\Big\{\opll''(t),\; -\opll(t)+\lam\int_{\chi(t)}^1 Q_{\rho}(1-s)\ds\Big\}=0, \quad \mbox{a.e. } t\in(0,1),\\
\opll(0+)=0, \quad\opll'(1-)=\opl'(\chi(1-))\chi'(1-)=c\alpha .
\end{cases}
\end{align}
This indicates that $\opll$ is the largest convex function dominated by the map $$t\mapsto \lam\int_{\chi(t)}^1 Q_{\rho}(1-s)\ds$$ on $(0,1)$ and satisfies the boundary conditions $\opll(0+)=0$ and $\opll'(1-)=c\alpha$.
This leads to an analytical expression of the solution to the problem \eqref{q2}: $\opl(t)=\opll(\chi^{-1}(t))$.

Because of the special form of the exponential utility function $u$, we can solve \eqref{q2} in an alternative way. Indeed, rewrite \eqref{q2} as
\begin{align*}
V_{\lam}=\sup_{Q\in \setq} \int_0^1\big[u\big(Q(t)\big) |u(\game(t))| -\lam Q(t)Q_{\rho}(1-t)\big]\dt.
\end{align*}
Thanks to the triangle inequality and \citeassmp{boundedclaim}, we have
\begin{align*}
\BE{|u(\pr)|}&\leq \BE{|u(\pr)-u(\game(0))|}+|u(\game(0))|\\
&=\BE{u(\pr)-u(\game(0))}+|u(\game(0))|<\infty.
\end{align*}
Also we may assume $\BE{|u(\pr)|}>0$; for otherwise $\pr$ would be a constant and \eqref{eut+} would reduce to the classical EUM problem \eqref{eut} that has been solved in the literature. The above allows us to define
\[ w(t):=\frac{1}{\BE{|u(\pr)|}}\int_0^{t} |u(\game(1-s))| \ds.\]
It is easy to check $w(0)=0$, $w(1)=1$ and $w$ is differentiable and increasing; so $w$ is a probability weighting (or distortion) function broadly involved in behavioral finance and insurance among others; see, e.g., Kahneman and Tversky's prospect theory \cite{KT79}, Quiggin's rank--dependent expected utility theory \cite{Q82} and Wang's distortion insurance premium principle \cite{W95,W96,WYP97}.

We now rewrite the problem \eqref{q2} as
\begin{align}\label{opi3b}
V_{\lam}=\BE{|u(\pr)|}\sup_{Q\in \setq} \int_0^1\Big[u\big(Q(t)\big)w'(1-t)-\lam Q(t)\frac{Q_{\rho}(1-t)}{\BE{|u(\pr)|}} \Big]\dt.
\end{align}
Intriguingly, the above is actually a quantile optimization problem arising from a behavioral finance model under the rank-dependent utility theory that has been studied and solved by Xia and Zhou \cite{XZ16} and Xu \cite{X16}.
\begin{theorem}\label{mainthm}
Let $u(x)=-e^{-\alpha x}$, $\alpha>0$, and let
$\delta(\cdot)$ be the concave envelope of the following function
\[t\mapsto -\int_0^{w^{-1}(1-t)}\frac{Q_{\rho}(s)}{\BE{|u(\pr)|}}\ds\]
on $[0,1]$. Then
\[\barq(t):=(u')^{-1}\big(\lam\delta'(t) \big),\quad t\in(0,1),\]
is an optimal solution to \eqref{opi3b} as well as to \eqref{q2}; and
\[X^*:=(u')^{-1}\left(\lam\delta'(1-U)\right)\]
is an optimal solution to the problem \eqref{static1}, where $U\in\setrho$ is arbitrary and $\lam>0$ satisfies
\[\int_0^1 (u')^{-1}\left(\lam\delta'(t) \right)Q_{\rho}(1-t)\dt=x.\]

\end{theorem}

Indeed, we have the following more explicit expression for $X^*$ without involving the Lagrange multiplier $\lam$.
\begin{corollary}\label{specialcase}
Let $U\in\setrho$ and
\begin{align*}
X^*:=\frac{1}{\BE{\rho}}\left(x+\alpha^{-1}\int_0^1 \ln(\delta'(s))Q_{\rho}(1-s)\ds\right)-\alpha^{-1} \ln(\delta'(1-U)).
\end{align*}	
Then $X^*$ is an optimal solution to \eqref{static1} with the initial endowment $x>0$.
\end{corollary}
\begin{proof}
By \citeprop{quantile},
it is sufficient to prove that the optimal solution $\barq$ to \eqref{opi3b} given in \citethm{mainthm} satisfies
\begin{align*}
\barq(t)=\frac{1}{\BE{\rho}}\left(x+\alpha^{-1}\int_0^1 \ln(\delta'(s))Q_{\rho}(1-s)\ds\right)-\alpha^{-1} \ln(\delta'(t)).
\end{align*}	
Recalling that $u(x)=-e^{-\alpha x}$, $u'(x)=\alpha e^{-\alpha x}$, we get
\[(u')^{-1}(x)=\alpha^{-1}\ln(\alpha)-\alpha^{-1}\ln(x),\quad x>0;\]
so
\begin{align*}
x&=\int_0^1 (u')^{-1}\left(\lam \delta'(s)\right)Q_{\rho}(1-s)\ds\\
&=\int_0^1 (\alpha^{-1}\ln(\alpha)-\alpha^{-1}\ln(\lam )-\alpha^{-1} \ln(\delta'(s)))Q_{\rho}(1-s)\ds\\
&=(\alpha^{-1}\ln(\alpha)-\alpha^{-1}\ln(\lam ))\BE{\rho}-\alpha^{-1}\int_0^1 \ln(\delta'(s))Q_{\rho}(1-s)\ds.
\end{align*}
Hence
\begin{align*}
\alpha^{-1}\ln(\alpha)-\alpha^{-1}\ln(\lam )=\frac{1}{\BE{\rho}}\left(x+\alpha^{-1}\int_0^1 \ln(\delta'(s))Q_{\rho}(1-s)\ds\right).
\end{align*}
Consequently,
\begin{align*}
\barq(t)=(u')^{-1}\left(\lam \delta'(t)\right)&=\alpha^{-1}\ln(\alpha)-\alpha^{-1}\ln(\lam )-\alpha^{-1} \ln(\delta'(t))\\
&=\frac{1}{\BE{\rho}}\left(x+\alpha^{-1}\int_0^1 \ln(\delta'(s))Q_{\rho}(1-s)\ds\right)-\alpha^{-1} \ln(\delta'(t)),
\end{align*}	
as desired.
\qed
\end{proof}

Now let us consider the optimal portfolio to the problem \eqref{eut+}.
To get more explicit expression, we assume all the market parameters are deterministic. In this case, $\rho$ is log-normally distributed (whose mean and variance can be given explicitly in terms of $r$ and $\theta$); thus $F_{\rho}$ is a continuous function and $\setrho=\{F_{\rho}(\rho)\}$.
In Theorem \ref{Optimalportfolio}, we can only take $U=F_{\rho}(\rho)$ and $X^{*}$ in Corollary \ref{specialcase} becomes
$X^{*}=\phi(\rho)$, where
\[\phi(y)=\frac{1}{\BE{\rho}}\left(x+\alpha^{-1}\int_0^1 \ln(\delta'(s))Q_{\rho}(1-s)\ds\right)-\alpha^{-1} \ln(\delta'(1-F_{\rho}(y))).\]

Now assume the following partial differential equation admits a $C^{1,2}$-smooth solution $\varphi$\footnote{This is guaranteed if both $\theta$ and $r$ are smooth functions of $t$ and $\theta(t)\neq 0$ for all $t\in[0,T]$. }:
\begin{align*}
\begin{cases}
\displaystyle\frac{\partial \varphi}{\partial t}+\frac{1}{2}|\theta(t)|^{2}y^{2}\frac{\partial^{2}\varphi}{\partial y^{2}}-
r(t)y\frac{\partial\varphi}{\partial y}=0, &\quad (t,y)\in[0,T) \times(0,\infty);\medskip\\
\displaystyle\varphi\big|_{(T,y)}=y\phi(y), &\quad y\in(0,\infty).
\end{cases}
\end{align*}
Let $\varrho$ be defined by \eqref{varrho}.
Applying It\^{o}'s lemma to $\varphi(t,\varrho(t))$, we obtain
\begin{align*}
\BE{\rho X^{*}\;\Big|\;\BF_t}
=\BE{\varphi(T,\varrho(T))\;\Big|\;\BF_t}=\varphi(t,\varrho(t))
=\varphi(0, \varrho(0))-\int_{0}^{t}\frac{\partial \varphi}{\partial y}(s,\varrho(s))\varrho(s)\theta(s)^{\tr}\dd W_{t}.
\end{align*}
Therefore, the process $Z$ in Theorem \ref{Optimalportfolio} is given by
\begin{align*}
Z(t)=-\frac{\partial \varphi}{\partial y}(t,\varrho(t))\varrho(t)\theta(t),\quad t\in[0,T].
\end{align*}
Substituting it into \eqref{optimalportfoliopi} gives an explicit feedback optimal portfolio to the problem \eqref{eut+}:
\begin{align*}
\pi^{*}(t)
&=\big(\sigma(t)^{\tr}\big)^{-1}\theta(t)\Big(X^{\pi^{*}}(t)-\frac{\partial \varphi}{\partial y}(t,\varrho(t))\Big),\quad t\in[0,T].
\end{align*}

\section{Concluding Remarks}\label{conclude}

This paper formulates and solves a robust EUM problem with an intractable claim by employing quantile formulation and calculus of variations, with the solution derived in terms of certain variational inequalities of ODEs. The result can be also used to {\it price}
the intractable claim $\pr$. Because $\pr$ is unpredictable and unhedgeable, we are unable to determine its (either supper or lower) hedging price the usual way. However, the results of this paper allow us to determine its price according to the famous utility indifference pricing principle; see e.g. \cite{CSW01,D97,MZ04a,MZ04b,OZ03,PTW09,RE00}.

Specifically, suppose the investor, with an initial endowment $x$, does not possess the intractable claim initially but is considering buying it. Then her EUM problem before purchasing $\pr$ is given by \eqref{eut}, with the optimal value $\veum(x)$. If she pays $p$ dollars to buy the claim, then her EUM problem after the purchase becomes \eqref{eut+} with $x$ replaced by $x-p$, leading to the optimal value $V_0(x-p)$. The utility indifference pricing theory stipulates that $p$ should be such that the investor is ``indifferent'' between buying and not buying $\pr$, namely, $p$ satisfies
\begin{align}\label{pricing}
\veum(x)=V_0(x-p).
\end{align}
Notice
\[\lim_{p\to-\infty}V_0(x-p)=u(\infty)>\veum(x),\]
and
\[V_0(0)=\BE{u(\pr)}.\]
It follows from \citelem{unboundedV} that the function $V_0$ is continuous and strictly increasing; so \eqref{pricing} admit a unique solution $p\leq x$ if and only if $\veum(x)\geq \BE{u(\pr)}$.
Note the price is negative if the claim is negative.

Beyond the utility indifference pricing,
our result can be possibly applied to price the intractable claim under other pricing principle; for instance, the marginal utility-based pricing in \cite{HKS05}.
We leave the details to interested readers.

\end{document}